\documentclass[11pt]{article} 

\usepackage[T1]{fontenc}        
\usepackage{textcomp, amssymb}  
\usepackage[english]{babel}

\usepackage{anysize}            
\usepackage{amsfonts}
\usepackage{amsmath}
\usepackage{dsfont}
\usepackage{MnSymbol}
\usepackage{mathtools}
\usepackage{epsfig}
\usepackage{epstopdf}
\usepackage{tikz}
\usepackage{amsthm}
\usepackage{ifpdf}

\usetikzlibrary{calc,decorations.pathmorphing,shapes,graphs}

\newcounter{sarrow}

\newcounter{sarrow1}
\newcommand\xnrsquigarrow[1]{%
\stepcounter{sarrow1}%
\mathrel{\begin{tikzpicture}[baseline= {( $ (current bounding box.south) + (0,-0.5ex) $ )}]
\node[inner sep=.5ex] (\thesarrow) {$\scriptstyle #1$};
\path[draw,<-,decorate,
  decoration={zigzag,amplitude=0.7pt,segment length=1.2mm,pre=lineto,pre length=4pt}]
    (\thesarrow1.south east) -- (\thesarrow1.south west);
    $\slashedarrowfill@\relbar\relbar/$
    \end{tikzpicture}}%
}

\makeatletter
\def\slashedarrowfill@#1#2#3#4#5{%
  $\m@th\thickmuskip0mu\medmuskip\thickmuskip\thinmuskip\thickmuskip
   \relax#5#1\mkern-7mu%
   \cleaders\hbox{$#5\mkern-2mu#2\mkern-2mu$}\hfill
   \mathclap{#3}\mathclap{#2}%
   \cleaders\hbox{$#5\mkern-2mu#2\mkern-2mu$}\hfill
   \mkern-7mu#4$%
}
\def\rightslashedarrowfillb@{%
  \slashedarrowfill@\relbar\relbar/\rightarrow}
\newcommand\xnrightarrow[2][]{%
  \ext@arrow 0055{\rightslashedarrowfillb@}{#1}{#2}}

\def\rightslashedarrowfille@{%
  \slashedarrowfill@\relbar\relbar/\twoheadrightarrow}
\newcommand\xntworightarrow[2][]{%
  \ext@arrow 0055{\rightslashedarrowfille@}{#1}{#2}}

\def\rightslashedarrowfillg@{%
  \slashedarrowfill@\relbar\relbar{\raisebox{.12em}{}}\twoheadrightarrow}
\newcommand\xtworightarrow[2][]{%
  \ext@arrow 0055{\rightslashedarrowfillg@}{#1}{#2}}

\def\rightslashedarrowfillx@{%
  \slashedarrowfill@\Relbar\Relbar/\rightrightarrows}
\newcommand\xnTworightarrow[2][]{%
  \ext@arrow 0055{\rightslashedarrowfillx@}{#1}{#2}}

\def\rightslashedarrowfilly@{%
  \slashedarrowfill@\Relbar\Relbar{\raisebox{.12em}{}}\rightrightarrows}
\newcommand\xTworightarrow[2][]{%
  \ext@arrow 0055{\rightslashedarrowfilly@}{#1}{#2}}

\pgfdeclareshape{slash underlined}
{
  \inheritsavedanchors[from=rectangle] 
  \inheritanchorborder[from=rectangle]
  \inheritanchor[from=rectangle]{north}
  \inheritanchor[from=rectangle]{north west}
  \inheritanchor[from=rectangle]{north east}
  \inheritanchor[from=rectangle]{center}
  \inheritanchor[from=rectangle]{west}
  \inheritanchor[from=rectangle]{east}
  \inheritanchor[from=rectangle]{mid}
  \inheritanchor[from=rectangle]{mid west}
  \inheritanchor[from=rectangle]{mid east}
  \inheritanchor[from=rectangle]{base}
  \inheritanchor[from=rectangle]{base west}
  \inheritanchor[from=rectangle]{base east}
  \inheritanchor[from=rectangle]{south}
  \inheritanchor[from=rectangle]{south west}
  \inheritanchor[from=rectangle]{south east}
  \inheritanchorborder[from=rectangle]
  \foregroundpath{
    \southwest \pgf@xa=\pgf@x \pgf@ya=\pgf@y
    \northeast \pgf@xb=\pgf@x \pgf@yb=\pgf@y
    \pgf@xc=\pgf@xa
    \advance\pgf@xc by .5\pgf@xb
    \pgf@yc=\pgf@ya
    \advance\pgf@xc by -1.3pt
    \advance\pgf@yc by -1.8pt
    \pgfpathmoveto{\pgfqpoint{\pgf@xc}{\pgf@yc}}
    \advance\pgf@xc by  2.6pt
    \advance\pgf@yc by  3.6pt
    \pgfpathlineto{\pgfqpoint{\pgf@xc}{\pgf@yc}}
    \pgfpathmoveto{\pgfqpoint{\pgf@xa}{\pgf@ya}}
    \pgfpathlineto{\pgfqpoint{\pgf@xb}{\pgf@ya}}
 }
}
\tikzset{nomorepostaction/.code=\let\tikz@postactions\pgfutil@empty}
\newcommand{\sembrack}[1]{[\![#1]\!]}

\providecommand{\fatbar}{\mkern1mu\mathpalette\v@rectangle\relax\mkern1mu}
\newcommand{\v@rectangle}[2]{%
  \hbox{
  \fboxrule=0.5\fontdimen 8
    \ifx#1\displaystyle\textfont\else
    \ifx#1\textstyle\textfont\else
    \ifx#1\scriptstyle\scriptfont\else
    \scriptscriptfont\fi\fi\fi 3
  \fboxsep=-\fboxrule
  \fbox{$\m@th#1\phantom{(}$}%
  }
}

\newtheorem{theorem}{Theorem}[section]
\newtheorem{definition}[theorem]{Definition}
\newtheorem{proposition}[theorem]{Proposition}
\newtheorem{lemma}[theorem]{Lemma}

\begin{document}

\begin{titlepage}
\thispagestyle{empty}

\hrule
\begin{center}
{\bf\LARGE Structured Parallel Programming\\}
%
\vspace{0.7cm}
--- Yong Wang ---

\vspace{2cm}
\begin{figure}[!htbp]
 \centering
 \includegraphics[width=1.0\textwidth]{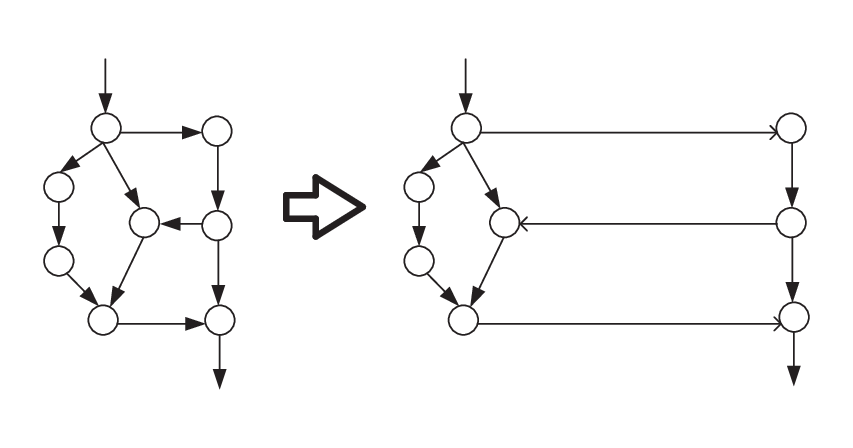}
\end{figure}

\end{center}
\end{titlepage}

\newpage 

\setcounter{page}{1}\pagenumbering{roman}

\tableofcontents

\newpage

\setcounter{page}{1}\pagenumbering{arabic}

        \section{Introduction}\label{intro}

Parallel computing \cite{GREEN} \cite{PEC} is becoming more and more important. Traditional parallelism often existed in distributed computing, since distributed systems are usually
autonomous and local computer is single-core and single-processor and timed (Timed computing is serial in nature). Today, due to the progress of hardware, multi-cores, multi-processors, and GPUs make the local computer truly parallel.

Parallel programming language has a relatively long research history. There have been always two ways: one is the structured way, and the other is the graph-based (true concurrent) way. The structured way is often based on the interleaving semantics, such as process algebra CCS. Since the parallelism in interleaving semantics is not a fundamental computational pattern (the parallel operator can be replaced by alternative composition and sequential composition), the parallel operator often does not occur as an explicit operator, such as the mainstream programming languages C, C++, Java, et al.

The graph-based way is also called true concurrency \cite{CM} \cite{ES1} \cite{ES2}. There also have been some ways to structure the graph \cite{DDP} \cite{SOP}, but these work only considered the causal relation in the graph, and neglected the confliction and even the communication. And there are also industrial efforts to adopt the graph-based way, such as the workflow description language WSFL. The later workflow description language BPEL adopts both the structured way and the graph-based way. Why does BPEL not adopt the structured way only? It is because that the expressive power of the structured way is limited. Then why does BPEL not adopt the graph-based way only? It is just because that the graph could not be structured at that time and the structured way is the basis on implementing a compiler.

We did some work on truly concurrent process algebra \cite{APTC}, which proved that truly concurrent process algebra is a generalization of traditional process algebra and had a side effect on the structuring true concurrency.

Now, it is the time to do some work on structured parallel programming under the background of programming language and parallel software engineering. On one side, traditional structured programming got great successes in sequential computation \cite{GOTO} \cite{SP}; on the other side, current structured parallel programming focused on parallel patterns (also known as parallel skeletons, templates, archetypes) \cite{PP0} \cite{PP1} \cite{PP2} \cite{PP3} \cite{PP4}, with comparison to structured sequential programming, the corresponding structured parallel programming with solid foundation still is missing.

In this paper, we try to clarify structured parallel programming corresponding to traditional structured sequential programming. This paper is organized as follows. In section \ref{pac}, we introduce the backgrounds of structured and unstructured parallelism. We introduce a parallel programming language called PPL in section \ref{ppl}.

\newpage\section{Parallelism and Concurrency}\label{pac}

In this section, we analyze the concepts of parallelism and concurrency, unstructured parallelism and structured parallelism.

We introduce unstructured parallelism in section \ref{uptc}, structured parallelism in section \ref{sp}, and the way from unstructured parallelism to structured parallelism in section \ref{upsp}. In section \ref{fusp}, we give the foundation of unstructured and structured parallel computation.

\subsection{Unstructured Parallelism - True Concurrency}\label{uptc}

True concurrency is usually defined by a graph-like structure \cite{ES1} \cite{ES2}, such as DAG (Directed Acyclic Graph), Petri net and event structure. As follows, we give the definition of Prime Event Structure.

\begin{definition}[Prime event structure]\label{PES}
Let $\Lambda$ be a fixed set of labels, ranged over $a,\cdots$. A ($\Lambda$-labelled) prime event structure is a tuple $\mathcal{E}=\langle \mathbb{E}, \leq, \sharp, \lambda\rangle$, where $\mathbb{E}$ is a denumerable set of events. Let $\lambda:\mathbb{E}\rightarrow\Lambda$ be a labelling function. And $\leq$, $\sharp$ are binary relations on $\mathbb{E}$, called causality and conflict respectively, such that:

\begin{enumerate}
  \item $\leq$ is a partial order and $\lceil e \rceil = \{e'\in \mathbb{E}|e'\leq e\}$ is finite for all $e\in \mathbb{E}$.
  \item $\sharp$ is irreflexive, symmetric and hereditary with respect to $\leq$, that is, for all $e,e',e''\in \mathbb{E}$, if $e\sharp e'\leq e''$, then $e\sharp e''$.
\end{enumerate}

Then, the concepts of consistency and concurrency can be drawn from the above definition:

\begin{enumerate}
  \item $e,e'\in \mathbb{E}$ are consistent, denoted as $e\frown e'$, if $\neg(e\sharp e')$. A subset $X\subseteq \mathbb{E}$ is called consistent, if $e\frown e'$ for all $e,e'\in X$.
  \item $e,e'\in \mathbb{E}$ are concurrent, denoted as $e\parallel e'$, if $\neg(e\leq e')$, $\neg(e'\leq e)$, and $\neg(e\sharp e')$.
\end{enumerate}
\end{definition}

In the Prime Event Structure defined true concurrency, we can see that there exist two kinds of unstructured relations: causality and confliction. Figure \ref{EOTC1} and Figure \ref{EOTC3} illustrates these two kinds of concurrency (for the simplicity, we separate the causal relation and the conflict relation).

Figure \ref{EOTC1} illustrates an example of primitives (atomic actions, events) with causal relations. Note that, primitives, atomic actions and events are almost the same concepts
under different backgrounds of computer science, and we will use them with no differences.

\begin{figure}
    \centering
    \includegraphics{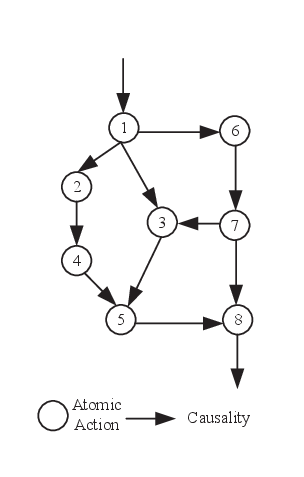}
    \caption{An example of unstructured parallelism}
    \label{EOTC1}
\end{figure}

Figure \ref{EOTC3} illustrates an example of atomic actions with causal relations and conflict relations. There exists a conflict relation between the second action in the left parallel branch and the second action in the right parallel branch, if the condition $b$ is $\mathbf{true}$, then the second action and its subsequent actions in the left branch can execute, else the second action and its subsequent actions in the right branch will execute.

\begin{figure}
    \centering
    \includegraphics{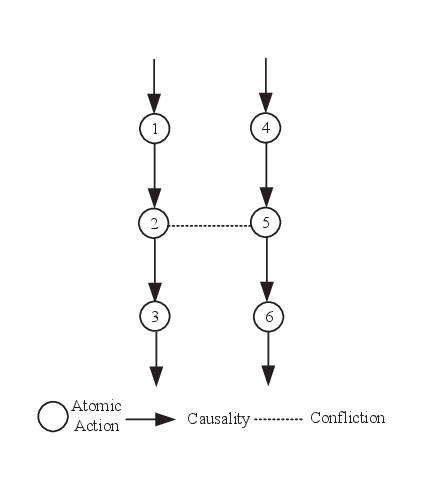}
    \caption{Another example of unstructured parallelism}
    \label{EOTC3}
\end{figure}

\subsection{Structured Parallelism}\label{sp}

Comparing to structured programming in sequential computation \cite{GOTO} \cite{SP}, we can intuitionally add a structured parallelism to the existed three basic programming structures (sequence, choice and iteration) of structured sequential programming, to form four basic programming structures of structured parallel programming: sequence, choice, iteration and parallelism. The intuitions and flow charts of the four basic structures are as follow.

The intuition of sequence (;) of two clauses $e_1;e_2$ is that after the successful execution of $e_1$, $e_2$ executes. The corresponding flow chat is shown in Figure \ref{FCOS}.

\begin{figure}
    \centering
    \includegraphics{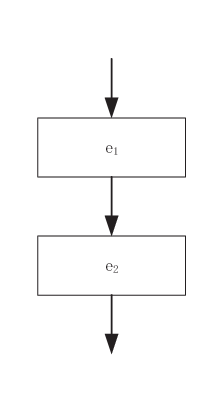}
    \caption{Sequence structure}
    \label{FCOS}
\end{figure}

The intuition of choice if $(b)$ then $e_1$ else $e_2$ is that if the condition $b$ is $\mathbf{true}$, then $e_1$ executes, else $e_2$ executes. The corresponding flow chat is shown in Figure \ref{FCOC}.

\begin{figure}
    \centering
    \includegraphics{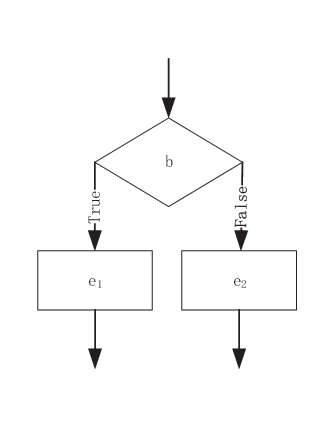}
    \caption{Choice structure}
    \label{FCOC}
\end{figure}

The intuition of iteration while $(b)$ do $e_1$ is that while the condition $b$ is $\mathbf{true}$, then $e_1$ executes many times. The corresponding flow chat is shown in Figure \ref{FCOI}.

\begin{figure}
    \centering
    \includegraphics{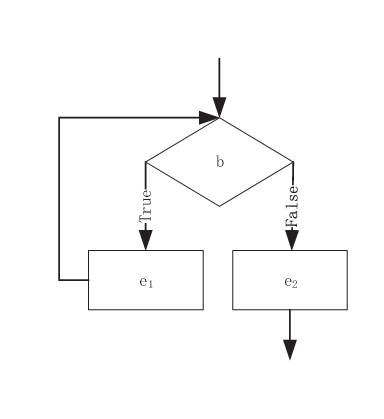}
    \caption{Iteration structure}
    \label{FCOI}
\end{figure}

The intuition of parallelism ($\parallel$) of two clauses $e_1\parallel e_2$ is that $e_1$ and $e_2$ execute simultaneously. The corresponding flow chat is shown in Figure \ref{FCOP}.

\begin{figure}
    \centering
    \includegraphics{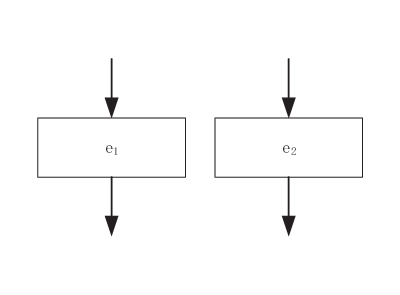}
    \caption{Parallelism structure}
    \label{FCOP}
\end{figure}

The programming of atomic actions, mixed by the above four structures is called structured parallel programming. We define Structured Parallel Program inductively as follows.

\begin{definition}[Structured parallel program]
Let the set of all primitives denote $\mathbb{P}$. A Structured Parallel Program $SPP$ is inductively defined as follows:

\begin{enumerate}
  \item $\mathbb{P}\subset SPP$;
  \item If $e_1\in SPP$ and $e_2\in SPP$, then $e_1;e_2\in SPP$;
  \item If $b$ is a condition, $e_1\in SPP$ and $e_2\in SPP$, then if $(b)$ then $e_1$ else $e_2\in SPP$;
  \item If $b$ is a condition, $e\in SPP$, then while $(b)$ do $e\in SPP$;
  \item If $e_1\in SPP$ and $e_2\in SPP$, then $e_1\parallel e_2\in SPP$.
\end{enumerate}
\end{definition}

\subsection{From Unstructured Parallelism to Structured Parallelism}\label{upsp}

The examples in Figure \ref{EOTC1} and \ref{EOTC3} are two kinds of typical unstructured parallelism. In this section, we try to structure these unstructured parallelism.

Firstly, the unstructured causalities in the same parallel branch can be structured by the famous conclusion that Goto statement is harmful \cite{GOTO} and also the similarly well-known structured (sequential) programming \cite{SP}; and for unstructured causalities, we find the example in Figure \ref{EOTC1} can not be structured, and the proof is stated in the following conclusions.

\begin{proposition}
The example in Figure \ref{EOTC1} can not be structured.
\end{proposition}

\begin{proof}
The actions 3 and 6 have the same causal pioneer 1, they should be in different parallel branches. But, the action 6 is the causal pioneer of the action 3 through the action 7, so, they should be in the same parallel branch. These cause contradictions.
\end{proof}

How can we deal this situation? Yes, we can classify the causal relations into two kinds: one is traditional sequential causality, and the other is the communication between different parallel branches, since the causality between parallel branches being communication is reasonable. Figure \ref{EOTC2} is the causality-classified one originated from Figure \ref{EOTC1}. This classification should be clarified during modelling time, that is, the programmer should draw Figure \ref{EOTC2} directly, instead of drawing Figure \ref{EOTC1} and then transforming it to Figure \ref{EOTC2}, in the modelling phase. Note that, multi-parties communications can be steadied by a series of two-parties communications without any loss.

\begin{figure}
    \centering
    \includegraphics{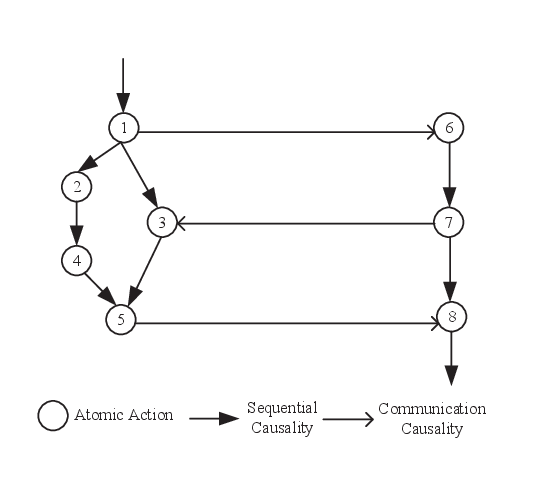}
    \caption{An example of structuring unstructured parallelism}
    \label{EOTC2}
\end{figure}

Then the causality-classified parallelism can be structured, we show the structuring way of synchronous and asynchronous communications.

For synchronous communication, The program corresponding to Figure \ref{EOTC1} can be written as follows:

$$(1;((2;4)\parallel 3);5)\parallel(6;7;8)$$

with three unstructured communications $sc_{1,6}$, $sc_{7,3}$ and $sc_{5,8}$.

The above program can be structured and equivalent to the following program:

$$sc_{1,6};((2;4)\parallel sc_{7,3});sc_{5,8}$$

We can see that the above program is structured, though the equivalence of the above two programs is not obvious. We will explain it through an rigorous way in the following chapters.

For asynchronous communication, the program corresponding to Figure \ref{EOTC1} can be written as follows:

$$(1;((2;4)\parallel 3);5)\parallel(6;7;8)$$

with three unstructured constraints $1\leq 6$, $7\leq 3$ and $5\leq 8$. Note that, $\leq$ is the causal relation.

The above program can be structured and equivalent to the following program:

$(1;((2;4)\parallel \textrm{if }(7\leq 3)\textrm{ then }3\textrm{ else }\mathbf{skip});5)
\parallel(\textrm{if }(1\leq 6)\textrm{ then }6\textrm{ else }\mathbf{skip};7;\textrm{if }(5\leq 8)\textrm{ then }8\textrm{ else }\mathbf{skip})$.

Note that $\mathbf{skip}$ is a voidness primitive.

The above conditions, like $1\leq 6$, $7\leq 3$ and $5\leq 8$, are not based on the the traditional results of data manipulation. Asynchronous communications are usually implemented by inserting an intermediate data structure, like mailbox or queue, between the two communicating partners, so, the above conditions can be the results of checking the data structure if the data are received in the data structure by the receiver. If the receiver has the ability to be blocked until the data are received, then the above conditions can be removed, and the structured program is the original one:

$$(1;((2;4)\parallel 3);5)\parallel(6;7;8)$$

without any constraint.


Then, it is turn to consider the unstructured conflictions between different parallel branches, since it is already proven that conflictions in the same parallel branch can be structured \cite{SP}, as the choice structure is a kind of structured confliction. Figure \ref{EOTC3} illustrates this kind of unstructured conflictions and can be expressed by the following program:

$$(1;2;3)\parallel (4;5;6)$$

with an unstructured confliction $2\sharp 5$, and a condition $b$, if $b$ is $\mathbf{true}$ then the primitive 2 and its successors execute, else the primitive 5 and its successors execute.

Figure \ref{EOTC3} can be structured by Figure \ref{EOTC4}. The structured program corresponding to Figure \ref{EOTC4} is:

$\textrm{if (b) then }(1;2;3)\parallel 4 \textrm{ else } 1\parallel (4;5;6)$


\begin{figure}
    \centering
    \includegraphics{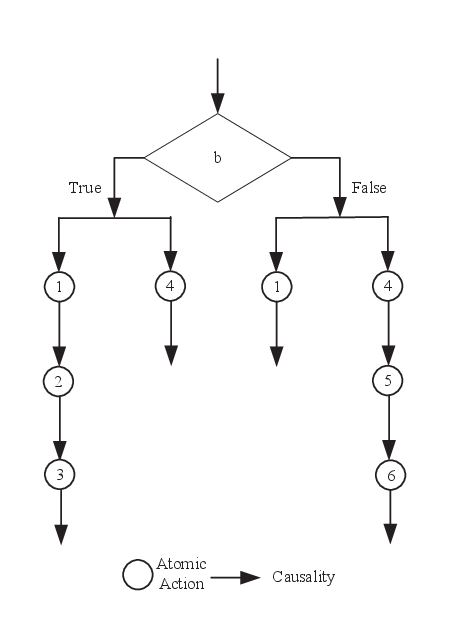}
    \caption{Another example of structuring unstructured parallelism}
    \label{EOTC4}
\end{figure}

\subsection{Foundation of Unstructured and Structured Parallelism}\label{fusp}

There existed several parallel machines \cite{PTM} \cite{CPTM} to provide the foundation for unstructured and structured parallelism since quite long time ago. Among them, the one (or multi)-tapes multi-heads Turing machine called PTM (Parallel Turing Machine) \cite{PTM} provides an intuitive foundation. The unstructured causalities and conflicts can be modelled as communications among the tape heads.

Prather \cite{STM} builded the so-called structured Turing machines with the four basic structures (sequence, choice, iteration and parallelism), which can realize every partial recursive function by a structured connection of simple machines.

\newpage\section{A Parallel Programming Language}\label{ppl}

In this section, we design a detailed parallel programming language, abbreviated PPL. PPL includes the four basic structures: sequence, choice, iteration and parallelism, and also non-determinism, communications (causalities between different parallel branches) and conflictions between different parallel branches. Note that, for the integrity, the semantics of traditional parts are also involved.

In section \ref{ppls}, we give the syntax of PPL. We give the operational semantics and denotational semantics in section \ref{pplos} and \ref{pplds}, and the relation between them in section \ref{pplrs}, we give the axiomatic semantics in section \ref{pplas}. We discuss non-determinism in section \ref{gc}, communications in section \ref{comm} and conflictions in section \ref{conf}, and the structuring algorithm in section \ref{sa}.

\subsection{Syntax}\label{ppls}

The syntactic sets of PPL are as follows.

\begin{itemize}
  \item Numbers set $\mathbf{N}$, with positive, negative integers and zero, and $n,m\in\mathbf{N}$;
  \item Truth values set $\mathbf{T}$, with values $\{\mathbf{true},\mathbf{false}\}$;
  \item Storage locations $\mathbf{Loc}$, and $X,Y\in\mathbf{Loc}$;
  \item Arithmetic expressions $\mathbf{Aexp}$, and $a\in\mathbf{Aexp}$;
  \item Boolean expressions $\mathbf{Bexp}$, and $b\in\mathbf{Bexp}$;
  \item Commands $\mathbf{Com}$, and $c\in\mathbf{Com}$.
\end{itemize}

The formation rules of PPL are:

For $\mathbf{Aexp}$:

$$a::=n\quad |\quad X\quad|\quad a_0+a_1\quad|\quad a_0-a_1\quad|\quad a_0\times a_1$$

For $\mathbf{Bexp}$:

$$b::=\mathbf{true}\quad |\quad \mathbf{false}\quad|\quad a_0=a_1\quad|\quad a_0\leq a_1\quad|\quad \neg b\quad|\quad b_0\wedge b_1\quad|\quad b_0\vee b_1$$

For $\mathbf{Com}$:

$$c::=\mathbf{skip}\quad |\quad X:=a\quad|\quad c_0;c_1\quad|\quad
\textbf{if } b\textbf{ then } c_0\textbf{ else } c_1\quad|\quad \textbf{while } b\textbf{ do } c\quad|\quad c_0\parallel c_1$$

We see that the syntax of PPL is almost same to traditional imperative language, except for the explicit parallel operator $\parallel$ in $\mathbf{Com}$.

\subsection{Operational Semantics}\label{pplos}

The set of states $\Sigma$ are composed of $\sigma:\mathbf{Loc}\rightarrow \mathbf{N}$, so, $\sigma(X)$ is the values of storage location $X$ under the state $\sigma$. For more about operational semantics, please refer to Plotkin's book \cite{SOS}.

In this section, we give the operational semantics of PPL.

\subsubsection{Operational Rules of $\mathbf{Aexp}$}

$\langle a,\sigma\rangle$ is called the configuration of arithmetic expression $a$, while $\langle a,\sigma\rangle\rightarrow n$ denotes that the value of $a$ is $n$ under the state $\sigma$.

The evaluation rule of integer $n$:

$$\langle n,\sigma\rangle\rightarrow n$$

The evaluation rule of storage location $X$:

$$\langle X,\sigma\rangle\rightarrow \sigma(X)$$

The evaluation rule of sums:

$$\frac{\langle a_0,\sigma\rangle\rightarrow n_0\quad\langle a_1,\sigma\rangle\rightarrow n_1}{\langle a_0+a_1,\sigma\rangle\rightarrow n},n=n_0+n_1$$

The evaluation rule of subtractions:

$$\frac{\langle a_0,\sigma\rangle\rightarrow n_0\quad\langle a_1,\sigma\rangle\rightarrow n_1}{\langle a_0-a_1,\sigma\rangle\rightarrow n},n=n_0-n_1$$

The evaluation rule of products:

$$\frac{\langle a_0,\sigma\rangle\rightarrow n_0\quad\langle a_1,\sigma\rangle\rightarrow n_1}{\langle a_0\times a_1,\sigma\rangle\rightarrow n},n=n_0\times n_1$$

Then we can define the following equivalence $\sim$ as follows.

\begin{definition}[Equivalence of operational semantics for arithmetic expressions]
$a_0\sim a_1$ iff $\forall n\in\mathbf{N}, \forall \sigma\in\Sigma.\langle a_0,\sigma\rangle\rightarrow n\Leftrightarrow\langle a_1,\sigma\rangle\rightarrow n$.
\end{definition}

\subsubsection{Operational Rules of $\mathbf{Bexp}$}

The evaluation rule of $\mathbf{true}$:

$$\langle \mathbf{true},\sigma\rangle\rightarrow \mathbf{true}$$

The evaluation rule of $\mathbf{false}$:

$$\langle \mathbf{false},\sigma\rangle\rightarrow \mathbf{false}$$

The evaluation rule of equality:

$$\frac{\langle a_0,\sigma\rangle\rightarrow n_0\quad\langle a_1,\sigma\rangle\rightarrow n_1}{\langle a_0=a_1,\sigma\rangle\rightarrow \mathbf{true}},n_0=n_1$$

$$\frac{\langle a_0,\sigma\rangle\rightarrow n_0\quad\langle a_1,\sigma\rangle\rightarrow n_1}{\langle a_0=a_1,\sigma\rangle\rightarrow \mathbf{false}},n_0\neq n_1$$

The evaluation rule of $\leq$:

$$\frac{\langle a_0,\sigma\rangle\rightarrow n_0\quad\langle a_1,\sigma\rangle\rightarrow n_1}{\langle a_0\leq a_1,\sigma\rangle\rightarrow \mathbf{true}},n_0\leq n_1$$

$$\frac{\langle a_0,\sigma\rangle\rightarrow n_0\quad\langle a_1,\sigma\rangle\rightarrow n_1}{\langle a_0\leq a_1,\sigma\rangle\rightarrow \mathbf{false}},n_0\geq n_1$$

The evaluation rule of $\neg$:

$$\frac{\langle b,\sigma\rangle\rightarrow \mathbf{true}}{\langle \neg b,\sigma\rangle\rightarrow \mathbf{false}}$$

$$\frac{\langle b,\sigma\rangle\rightarrow \mathbf{false}}{\langle \neg b,\sigma\rangle\rightarrow \mathbf{true}}$$

The evaluation rule of $\wedge$:

$$\frac{\langle b_0,\sigma\rangle\rightarrow t_0\quad\langle b_1,\sigma\rangle\rightarrow t_1}{\langle b_0\wedge b_1,\sigma\rangle\rightarrow t},t=\mathbf{true}, t_0\equiv\mathbf{true}\wedge t_1\equiv\mathbf{true};t=\mathbf{false},otherwise$$

The evaluation rule of $\vee$:

$$\frac{\langle b_0,\sigma\rangle\rightarrow t_0\quad\langle b_1,\sigma\rangle\rightarrow t_1}{\langle b_0\vee b_1,\sigma\rangle\rightarrow t},t=\mathbf{true}, t_0\equiv\mathbf{true}\vee t_1\equiv\mathbf{true};t=\mathbf{false},otherwise$$

Then we can define the following equivalence $\sim$ as follows.

\begin{definition}[Equivalence of operational semantics for boolean expressions]
$b_0\sim b_1$ iff $\forall t\in\mathbf{T}, \forall \sigma\in\Sigma.\langle b_0,\sigma\rangle\rightarrow t\Leftrightarrow\langle b_1,\sigma\rangle\rightarrow t$.
\end{definition}

\subsubsection{Operational Rules for $\mathbf{Com}$}

$\langle c,\sigma\rangle$ denotes the configuration of the command $c$, which means that the command $c$ executes under the state $\sigma$. And $\langle c\sigma\rangle\rightarrow \sigma'$ means that the command $c$ executing under the state $\sigma$ evolves to the state $\sigma'$. For $n\in\mathbf{N}$ and $X\in\mathbf{Loc}$, $\sigma[n/X]$ denotes using $n$ to replace the contents of $X$ under the state $\sigma$.

The execution rule of $\mathbf{skip}$:

$$\langle \mathbf{skip},\sigma\rangle\rightarrow \sigma$$

The execution rule of assignment:

$$\frac{\langle a,\sigma\rangle\rightarrow n}{\langle X:=a,\sigma\rangle\rightarrow \sigma[n/X]}$$

The execution rule of sequence:

$$\frac{\langle c_0,\sigma\rangle\rightarrow \sigma'}{\langle c_0;c_1,\sigma\rangle\rightarrow \langle c_1,\sigma'\rangle}$$

The execution rule of choice:

$$\frac{\langle b,\sigma\rangle \rightarrow \mathbf{true}\quad \langle c_0,\sigma\rangle\rightarrow \sigma'}{\langle \textbf{if }b\textbf{ then } c_0\textbf{ else }c_1,\sigma\rangle\rightarrow\sigma'}$$

$$\frac{\langle b,\sigma\rangle \rightarrow \mathbf{false}\quad \langle c_1,\sigma\rangle\rightarrow \sigma'}{\langle \textbf{if }b\textbf{ then } c_0\textbf{ else }c_1,\sigma\rangle\rightarrow\sigma'}$$

The execution rule of iteration:

$$\frac{\langle b,\sigma\rangle\rightarrow\mathbf{false}}{\langle\textbf{while }b\textbf{ do }c,\sigma\rangle\rightarrow \sigma}$$

$$\frac{\langle b,\sigma\rangle\rightarrow\mathbf{true}\quad \langle c,\sigma\rangle\rightarrow\sigma''\quad\langle \textbf{while }b\textbf{ do }c,\sigma''\rangle\rightarrow\sigma'}{\langle\textbf{while }b\textbf{ do }c,\sigma\rangle\rightarrow \sigma'}$$

The execution rule of parallelism:

%

$$\frac{\langle c_1,\sigma\rangle\rightarrow\sigma'\quad \langle c_0,\sigma\rangle\rightarrow\sigma''}{\langle c_0\parallel c_1,\sigma\rangle\rightarrow\sigma'\uplus\sigma''}$$

where $\sigma'\uplus\sigma''$ is the final states after $c_0$ and $c_1$ execute simultaneously.


Note that, for true concurrency, there are still three other properties should be processed: communications, conflictions, and race conditions (we leave them to the next section).

\begin{enumerate}
  \item Communication is occurring between two atomic communicating commands, which can be defined by a communication function $\gamma(c_0, c_1)\triangleq c(c_0,c_1)$. Communications can be implemented by several ways: share storage locations, invocation of functions by values, or network communications. For a pure imperative programming lange, we only consider the case of share storage locations, so, there is no need to define new communicating commands. So, two commands in communication is with a relation
  $\gamma(c_0, c_1)\triangleq c(c_0,c_1)$, but rules of $c_0\parallel c_1$ are still the same to the above ones. We will discuss the general communications in section \ref{comm};
  \item Confliction may have two forms: one exists as the condition rules define; the other may exist among the parallel branches, which must be eliminated. But the elimination of confliction existing in parallel branches may lead to non-deterministic results (refer to \cite{APTC} for details). For simplicity, we assume that the programs written by PPL at this time have no conflictions, because a program with the conflictions existing among parallel branches have an equal program without conflicts. That is, the conflictions can be eliminated and structured, and we will discuss the elimination of conflictions between parallel branches in section \ref{conf};
  \item Race condition may exist in two parallel commands, for example, they are all executing assignment to a same storage location. Two parallel commands in race condition must be executed serially. We should define new rules for race condition, but, these rules also lead to non-deterministic results. So, we also assume that the programs written by PPL deal with this situation and the non-deterministic execution is eliminated. In fact, we can write $c_0\parallel (\mathbf{skip};c_1)$ or $(\mathbf{skip};c_0)\parallel c_1$, or put $c_0,c_1$ in a condition, where $c_0$ and $c_1$ are in race condition. But, indeed, the above parallelism is still can be used widely in non-sharing memory computation (distributed computing), or non-racing of sharing memory computation. For the general form of non-determinism, we will discuss in section \ref{gc} and \ref{comm}.
\end{enumerate}

We can get the following propositions. Where $\sim$ is an equivalence relation on commands by the definition, where $\Sigma$ is the set of states:

\begin{definition}[Equivalence of operational semantics for commands]
$c_0\sim c_1\textrm{ iff } \forall\sigma,\sigma'\in\Sigma, \langle c_0,\sigma\rangle\rightarrow \sigma'\Leftrightarrow\langle c_1,\sigma\rangle\rightarrow \sigma'$
\end{definition}

\begin{proposition}
$c_0\parallel c_1\sim c_1\parallel c_0$, for $c_0,c_1\in\mathbf{Com}$.
\end{proposition}

\begin{proof}
By use of the transition rules of $\parallel$, we can get the following derivations of $c_0\parallel c_1$ for $\forall \sigma\in\Sigma$:

%

$$\frac{\langle c_1,\sigma\rangle\rightarrow\sigma'\quad \langle c_0,\sigma\rangle\rightarrow\sigma''}{\langle c_0\parallel c_1,\sigma\rangle\rightarrow\sigma'\uplus\sigma''}$$

And we can get the following derivations of $c_1\parallel c_0$ for $\forall \sigma\in\Sigma$:

%

$$\frac{\langle c_0,\sigma\rangle\rightarrow\sigma'\quad \langle c_1,\sigma\rangle\rightarrow\sigma''}{\langle c_1\parallel c_0,\sigma\rangle\rightarrow\sigma'\uplus\sigma''}$$

So, it is obvious that $c_0\parallel c_1\sim c_1\parallel c_0$, for $c_0,c_1\in\mathbf{Com}$, as desired.
\end{proof}

\begin{proposition}
$(c_0\parallel c_1)\parallel c_2\sim c_0\parallel (c_1\parallel c_2)$, for $c_0,c_1,c_2\in\mathbf{Com}$.
\end{proposition}

\begin{proof}
By use of the transition rules of $\parallel$, we can get the following derivations of $(c_0\parallel c_1)\parallel c_2$ for $\forall \sigma\in\Sigma$:

$$\frac{\langle c_0,\sigma\rangle\rightarrow\sigma'\quad \langle c_1,\sigma\rangle\rightarrow\sigma''\quad \langle c_2,\sigma\rangle\rightarrow\sigma'''}{\langle (c_0\parallel c_1)\parallel c_2,\sigma\rangle\rightarrow\sigma'\uplus\sigma''\uplus\sigma'''}$$

And we can get the following derivations of $c_0\parallel (c_1\parallel c_2)$ for $\forall \sigma\in\Sigma$:

$$\frac{\langle c_0,\sigma\rangle\rightarrow\sigma'\quad \langle c_1,\sigma\rangle\rightarrow\sigma''\quad \langle c_2,\sigma\rangle\rightarrow\sigma'''}{\langle c_0\parallel (c_1\parallel c_2),\sigma\rangle\rightarrow\sigma'\uplus\sigma''\uplus\sigma'''}$$

So, it is obvious that $(c_0\parallel c_1)\parallel c_2\sim c_0\parallel (c_1\parallel c_2)$, for $c_0,c_1,c_2\in\mathbf{Com}$, as desired.
\end{proof}

\begin{proposition}
$(\textbf{if } b\textbf{ then } c_0\textbf{ else } c_1)\parallel c_2 \sim \textbf{if } b\textbf{ then } c_0\parallel c_2\textbf{ else } c_1\parallel c_2$, for $c_0,c_1,c_2\in\mathbf{Com}$
\end{proposition}

\begin{proof}
By use of the transition rules of choice and $\parallel$, we can get the following derivations of $(\textbf{if } b\textbf{ then } c_0\textbf{ else } c_1)\parallel c_2$:

%

$$\frac{\langle b,\sigma\rangle \rightarrow \mathbf{true}\quad \langle c_0,\sigma\rangle\rightarrow \sigma'\quad \langle c_2,\sigma\rangle \rightarrow \sigma''}{\langle (\textbf{if }b\textbf{ then } c_0\textbf{ else }c_1)\parallel c_2,\sigma\rangle\rightarrow\sigma'\uplus\sigma''}$$

%

$$\frac{\langle b,\sigma\rangle \rightarrow \mathbf{false}\quad \langle c_1,\sigma\rangle\rightarrow \sigma'\quad \langle c_2,\sigma\rangle \rightarrow \sigma''}{\langle (\textbf{if }b\textbf{ then } c_0\textbf{ else }c_1)\parallel c_2,\sigma\rangle\rightarrow\sigma'\uplus\sigma''}$$

And we can get the following derivations of $\textbf{if } b\textbf{ then } c_0\parallel c_2\textbf{ else } c_1\parallel c_2$:

%

$$\frac{\langle b,\sigma\rangle \rightarrow \mathbf{true}\quad \langle c_0,\sigma\rangle\rightarrow \sigma'\quad \langle c_2,\sigma\rangle \rightarrow \sigma''}{\langle (\textbf{if } b\textbf{ then } c_0\parallel c_2\textbf{ else } c_1\parallel c_2,\sigma\rangle\rightarrow\sigma'\uplus\sigma''}$$

%

$$\frac{\langle b,\sigma\rangle \rightarrow \mathbf{false}\quad \langle c_1,\sigma\rangle\rightarrow \sigma'\quad \langle c_2,\sigma\rangle \rightarrow \sigma''}{\langle (\textbf{if } b\textbf{ then } c_0\parallel c_2\textbf{ else } c_1\parallel c_2,\sigma\rangle\rightarrow\sigma'\uplus\sigma''}$$

So, it is obvious that $(\textbf{if } b\textbf{ then } c_0\textbf{ else } c_1)\parallel c_2 \sim \textbf{if } b\textbf{ then } c_0\parallel c_2\textbf{ else } c_1\parallel c_2$, for $c_0,c_1,c_2\in\mathbf{Com}$, as desired.
\end{proof}

\begin{proposition}\label{SePa}
For $c_0,c_1,c_2,c_3\in\mathbf{Com}$,

\begin{enumerate}
  \item $(c_0;c_1)\parallel c_2\sim (c_0\parallel c_2);c_1$;
  \item $(c_0;c_1)\parallel (c_2;c_3)\sim (c_0\parallel c_2);(c_1\parallel c_3)$.
\end{enumerate}

\end{proposition}

\begin{proof}
(1) By use of the transition rules of sequence and $\parallel$, we can get the following derivations of $(c_0;c_1)\parallel c_2$:

$$\frac{\langle c_0,\sigma\rangle\rightarrow\sigma'\quad \langle c_2,\sigma\rangle\rightarrow\sigma''}{\langle (c_0;c_1)\parallel c_2,\sigma\rangle\rightarrow\langle c_1,\sigma'\uplus\sigma''\rangle}$$

And we can get the following derivations of $(c_0\parallel c_2);c_1$:

$$\frac{\langle c_0,\sigma\rangle\rightarrow\sigma'\quad \langle c_2,\sigma\rangle\rightarrow\sigma''}{\langle (c_0\parallel c_2);c_1,\sigma\rangle\rightarrow\langle c_1,\sigma'\uplus\sigma''\rangle}$$

So, it is obvious that $(c_0;c_1)\parallel c_2\sim (c_0\parallel c_2);c_1$, for $c_0,c_1,c_2\in\mathbf{Com}$, as desired.

(2) By use of the transition rules of sequence and $\parallel$, we can get the following derivations of $(c_0;c_1)\parallel (c_2;c_3)$:

$$\frac{\langle c_0,\sigma\rangle\rightarrow\sigma'\quad \langle c_2,\sigma\rangle\rightarrow\sigma''}{\langle (c_0;c_1)\parallel (c_2;c_3),\sigma\rangle\rightarrow\langle c_1\parallel c_3,\sigma'\uplus\sigma''\rangle}$$

And we can get the following derivations of $(c_0\parallel c_2);(c_1\parallel c_3)$:

$$\frac{\langle c_0,\sigma\rangle\rightarrow\sigma'\quad \langle c_2,\sigma\rangle\rightarrow\sigma''}{\langle (c_0\parallel c_2);(c_1\parallel c_3),\sigma\rangle\rightarrow\langle c_1\parallel c_3,\sigma'\uplus\sigma''\rangle}$$

So, it is obvious that $(c_0;c_1)\parallel (c_2;c_3)\sim (c_0\parallel c_2);(c_1\parallel c_3)$, for $c_0,c_1,c_2,c_3\in\mathbf{Com}$, as desired.
\end{proof}

\begin{proposition}\label{skipP}
$c\parallel \mathbf{skip}\sim c$, for $c\in\mathbf{Com}$.
\end{proposition}

\begin{proof}
By use of the transition rules of $\mathbf{skip}$ and $\parallel$, we can get the following derivations of $c\parallel \mathbf{skip}$:

$$\frac{\langle c,\sigma\rangle\rightarrow\sigma'\quad \langle \mathbf{skip},\sigma\rangle\rightarrow\sigma}{c\parallel \mathbf{skip},\sigma\rangle\rightarrow\sigma'\uplus\sigma}$$

And it is obvious that:

$$\frac{\langle c,\sigma\rangle\rightarrow\sigma'}{c,\sigma\rangle\rightarrow\sigma'}$$

For $\sigma'\uplus\sigma=\sigma'$, it is obvious that $c\parallel \mathbf{skip}\sim c$, for $c\in\mathbf{Com}$, as desired.
\end{proof}

\begin{lemma}\label{skipP2}
For  $c_0,c_1\in\mathbf{Com}$,

\begin{enumerate}
  \item $c_0\parallel c_1\sim c_0\parallel (\mathbf{skip};c_1)\sim c_0;c_1$;
  \item $c_0\parallel c_1\sim (\mathbf{skip};c_0)\parallel c_1\sim c_1;c_0$.
\end{enumerate}
\end{lemma}

\begin{proof}
It is obvious by Proposition \ref{SePa} and \ref{skipP}.
\end{proof}

From Lemma \ref{skipP2}, we can see that the execution orders of $c_0\parallel c_1$ cause non-determinism, they can be executed in any sequential order or in parallel simultaneously. But, without race condition, the final states after the execution of $c_0\parallel c_1$ are deterministic.

\subsection{Denotational Semantics}\label{pplds}

Denotational semantics can be used to describe the semantics of PPL. For more about denotational semantics, please refer to Mosses's book \cite{DS}.

In this section, we give the denotational semantics for PPL.

\subsubsection{Denotational Semantics of $\mathbf{Aexp}$}

We define the denotational semantics of $\mathbf{Aexp}$ as $\mathcal{A}:\mathbf{Aexp}\rightarrow(\Sigma\rightarrow \mathbf{N})$. The concrete denotational semantics of $\mathbf{Aexp}$ are following.

$\mathcal{A}\sembrack{n}=\{(\sigma,n)|\sigma\in\Sigma\}$

$\mathcal{A}\sembrack{X}=\{(\sigma,\sigma(X))|\sigma\in\Sigma\}$

$\mathcal{A}\sembrack{a_0+a_1}=\{(\sigma,n_0+n_1)|(\sigma,n_0)\in\mathcal{A}\sembrack{a_0}\&(\sigma,n_1)\in\mathcal{A}\sembrack{a_1}\}$

$\mathcal{A}\sembrack{a_0-a_1}=\{(\sigma,n_0-n_1)|(\sigma,n_0)\in\mathcal{A}\sembrack{a_0}\&(\sigma,n_1)\in\mathcal{A}\sembrack{a_1}\}$

$\mathcal{A}\sembrack{a_0\times a_1}=\{(\sigma,n_0\times n_1)|(\sigma,n_0)\in\mathcal{A}\sembrack{a_0}\&(\sigma,n_1)\in\mathcal{A}\sembrack{a_1}\}$

\subsubsection{Denotational Semantics of $\mathbf{Bexp}$}

We define the denotational semantics of $\mathbf{Bexp}$ as $\mathcal{B}:\mathbf{Bexp}\rightarrow(\Sigma\rightarrow \mathbf{T})$. The concrete denotational semantics of $\mathbf{Bexp}$
are following.

$\mathcal{B}\sembrack{\mathbf{true}}=\{(\sigma,\mathbf{true})|\sigma\in\Sigma\}$

$\mathcal{B}\sembrack{\mathbf{false}}=\{(\sigma,\mathbf{false})|\sigma\in\Sigma\}$

$\mathcal{B}\sembrack{a_0=a_1}=\{(\sigma,\mathbf{true})|\sigma\in\Sigma\&\mathcal{A}\sembrack{a_0}\sigma=\mathcal{A}\sembrack{a_1}\sigma\}\cup
\{(\sigma,\mathbf{false})|\sigma\in\Sigma\&\mathcal{A}\sembrack{a_0}\sigma\neq\mathcal{A}\sembrack{a_1}\sigma\}$

$\mathcal{B}\sembrack{a_0\leq a_1}=\{(\sigma,\mathbf{true})|\sigma\in\Sigma\&\mathcal{A}\sembrack{a_0}\sigma\leq\mathcal{A}\sembrack{a_1}\sigma\}\cup
\{(\sigma,\mathbf{false})|\sigma\in\Sigma\&\mathcal{A}\sembrack{a_0}\sigma\nleq\mathcal{A}\sembrack{a_1}\sigma\}$

$\mathcal{B}\sembrack{\neg b}=\{(\sigma,\neg_Tt)|\sigma\in\Sigma\&(\sigma,t)\in\mathcal{B}\sembrack{b}\}$

$\mathcal{B}\sembrack{b_0\wedge b_1}=\{(\sigma,t_0\wedge_T t_1)|\sigma\in\Sigma\&(\sigma,t_0)\in\mathcal{B}\sembrack{b_0}\&(\sigma,t_1)\in\mathcal{B}\sembrack{b_1}\}$

$\mathcal{B}\sembrack{b_0\vee b_1}=\{(\sigma,t_0\vee_T t_1)|\sigma\in\Sigma\&(\sigma,t_0)\in\mathcal{B}\sembrack{b_0}\&(\sigma,t_1)\in\mathcal{B}\sembrack{b_1}\}$

\subsubsection{Denotational Semantics of $\mathbf{Com}$}

We define the denotational semantics of $\mathbf{Com}$ as $\mathcal{C}:\mathbf{Com}\rightarrow(\Sigma\rightarrow \Sigma)$. The denotational semantics of $\mathbf{Com}$ are following.

$\mathcal{C}\sembrack{\mathbf{skip}}=\{(\sigma,\sigma)|\sigma\in\Sigma\}$

$\mathcal{C}\sembrack{X:=a}=\{(\sigma,\sigma[n/X])|\sigma\in\Sigma\&n=\mathcal{A}\sembrack{a}\sigma\}$

$\mathcal{C}\sembrack{c_0;c_1}=\mathcal{C}\sembrack{c_1}\circ\mathcal{C}\sembrack{c_0}$

$\mathcal{C}\sembrack{\textrm{if }b\textrm{ then }c_0\textrm{ else}c_1}=\{(\sigma,\sigma')|\mathcal{B}\sembrack{b}\sigma=\mathbf{true}\&(\sigma,\sigma')\in\mathcal{C}\sembrack{c_0}\}\cup\\
\{(\sigma,\sigma')|\mathcal{B}\sembrack{b}\sigma=\mathbf{false}\&(\sigma,\sigma')\in\mathcal{C}\sembrack{c_1}\}$

$\mathcal{C}\sembrack{\textrm{while }b\textrm{ do }c}=fix(\Gamma)$

with $\Gamma(\phi)=\{(\sigma,\sigma')|\mathcal{B}\sembrack{b}\sigma=\mathbf{true}\&(\sigma,\sigma')\in\phi\circ\mathcal{C}\sembrack{c}\}\cup\\
\{(\sigma,\sigma')|\mathcal{B}\sembrack{b}\sigma=\mathbf{false}\}$

$\mathcal{C}\sembrack{c_0\parallel c_1}=\mathcal{C}\sembrack{c_0}\}\cup \{\mathcal{C}\sembrack{c_1}$

We can get the following propositions.

\begin{proposition}
$\mathcal{C}\sembrack{c_0\parallel c_1}=\mathcal{C}\sembrack{c_1\parallel c_0}$, for $c_0,c_1\in\mathbf{Com}$.
\end{proposition}

\begin{proof}
By the definition of the denotation of $\parallel$, we can get:

$\mathcal{C}\sembrack{c_0\parallel c_1}=\mathcal{C}\sembrack{c_0}\cup \mathcal{C}\sembrack{c_1}$

$\mathcal{C}\sembrack{c_1\parallel c_0}=\mathcal{C}\sembrack{c_1}\cup \mathcal{C}\sembrack{c_0}$

So, $\mathcal{C}\sembrack{c_0\parallel c_1}=\mathcal{C}\sembrack{c_1\parallel c_0}$, for $c_0,c_1\in\mathbf{Com}$, as desired.
\end{proof}

\begin{proposition}
$\mathcal{C}\sembrack{(c_0\parallel c_1)\parallel c_2}=\mathcal{C}\sembrack{c_0\parallel (c_1\parallel c_2)}$, for $c_0,c_1,c_2\in\mathbf{Com}$.
\end{proposition}

\begin{proof}
By the definition of the denotation of $\parallel$, we can get:

$\mathcal{C}\sembrack{(c_0\parallel c_1)\parallel c_2}=(\mathcal{C}\sembrack{c_0}\cup \mathcal{C}\sembrack{c_1})\cup \mathcal{C}\sembrack{c_2}$

$\mathcal{C}\sembrack{c_0\parallel (c_1\parallel c_2)}=\mathcal{C}\sembrack{c_0}\cup (\mathcal{C}\sembrack{c_1}\cup \mathcal{C}\sembrack{c_2})$

$\mathcal{C}\sembrack{(c_0\parallel c_1)\parallel c_2}=\mathcal{C}\sembrack{c_0\parallel (c_1\parallel c_2)}$, for $c_0,c_1,c_2\in\mathbf{Com}$, as desired.
\end{proof}

\begin{proposition}
$\mathcal{C}\sembrack{(\textbf{if } b\textbf{ then } c_0\textbf{ else } c_1)\parallel c_2}=\mathcal{C}\sembrack{\textbf{if } b\textbf{ then } c_0\parallel c_2\textbf{ else } c_1\parallel c_2}$, for $c_0,c_1,c_2\in\mathbf{Com}$
\end{proposition}

\begin{proof}
By the definition of the denotation of choice and $\parallel$, we can get:

$\mathcal{C}\sembrack{(\textbf{if } b\textbf{ then } c_0\textbf{ else } c_1)\parallel c_2}=\{(\sigma,\sigma')|\mathcal{B}\sembrack{b}\sigma=\mathbf{true}\&(\sigma,\sigma')\in\mathcal{C}\sembrack{c_0}\}\cup\\
\{(\sigma,\sigma')|\mathcal{B}\sembrack{b}\sigma=\mathbf{false}\&(\sigma,\sigma')\in\mathcal{C}\sembrack{c_1}\}\cup \mathcal{C}\sembrack{c_2}$

$\mathcal{C}\sembrack{\textbf{if } b\textbf{ then } c_0\parallel c_2\textbf{ else } c_1\parallel c_2}=\{(\sigma,\sigma')|\mathcal{B}\sembrack{b}\sigma=\mathbf{true}\&(\sigma,\sigma')\in\mathcal{C}\sembrack{c_0}\cup \mathcal{C}\sembrack{c_2}\}\cup\\
\{(\sigma,\sigma')|\mathcal{B}\sembrack{b}\sigma=\mathbf{false}\&(\sigma,\sigma')\in\mathcal{C}\sembrack{c_1}\cup \mathcal{C}\sembrack{c_2}\}$

So, $\mathcal{C}\sembrack{(\textbf{if } b\textbf{ then } c_0\textbf{ else } c_1)\parallel c_2}=\mathcal{C}\sembrack{\textbf{if } b\textbf{ then } c_0\parallel c_2\textbf{ else } c_1\parallel c_2}$, for $c_0,c_1,c_2\in\mathbf{Com}$, as desired.
\end{proof}

\begin{proposition}\label{SePa2}
For $c_0,c_1,c_2,c_3\in\mathbf{Com}$,

\begin{enumerate}
  \item $\mathcal{C}\sembrack{(c_0;c_1)\parallel c_2}=\mathcal{C}\sembrack{(c_0\parallel c_2);c_1}$;
  \item $\mathcal{C}\sembrack{(c_0;c_1)\parallel (c_2;c_3)}=\mathcal{C}\sembrack{(c_0\parallel c_2);(c_1\parallel c_3)}$.
\end{enumerate}

\end{proposition}

\begin{proof}
(1)By the definition of the denotation of sequence and $\parallel$, we can get:

$\mathcal{C}\sembrack{(c_0;c_1)\parallel c_2}=(\mathcal{C}\sembrack{c_1}\circ \mathcal{C}\sembrack{c_0})\cup \mathcal{C}\sembrack{c_2}$

$\mathcal{C}\sembrack{(c_0\parallel c_2);c_1}=\mathcal{C}\sembrack{c_1}\circ(\mathcal{C}\sembrack{c_0}\cup \mathcal{C}\sembrack{c_2})$

So, $\mathcal{C}\sembrack{(c_0;c_1)\parallel c_2}=\mathcal{C}\sembrack{(c_0\parallel c_2);c_1}$, as desired.

(1)By the definition of the denotation of sequence and $\parallel$, we can get:

$\mathcal{C}\sembrack{(c_0;c_1)\parallel (c_2;c_3)}=(\mathcal{C}\sembrack{c_1}\circ \mathcal{C}\sembrack{c_0})\cup (\mathcal{C}\sembrack{c_3}\circ\mathcal{C}\sembrack{c_2}$

$\mathcal{C}\sembrack{(c_0\parallel c_2);(c_1\parallel c_3)}=(\mathcal{C}\sembrack{c_1}\cup \mathcal{C}\sembrack{c_3})\circ (\mathcal{C}\sembrack{c_2}\cup\mathcal{C}\sembrack{c_0}$

So, $\mathcal{C}\sembrack{(c_0;c_1)\parallel (c_2;c_3)}=\mathcal{C}\sembrack{(c_0\parallel c_2);(c_1\parallel c_3)}$, as desired.
\end{proof}

\begin{proposition}\label{skipP2}
$\mathcal{C}\sembrack{c\parallel \mathbf{skip}}=\mathcal{C}\sembrack{c}$, for $c\in\mathbf{Com}$.
\end{proposition}

\begin{proof}
By the definition of the denotation of $\mathbf{skip}$ and $\parallel$, we can get:

$\mathcal{C}\sembrack{c\parallel \mathbf{skip}}=\mathcal{C}\sembrack{c}\cup \mathcal{C}\sembrack{\mathbf{skip}}$

So, $\mathcal{C}\sembrack{c\parallel \mathbf{skip}}=\mathcal{C}\sembrack{c}$, for $c\in\mathbf{Com}$, as desired.
\end{proof}

\begin{lemma}\label{skipP2}
For  $c_0,c_1\in\mathbf{Com}$,

\begin{enumerate}
  \item $c_0\parallel c_1\sim c_0\parallel (\mathbf{skip};c_1)\sim c_0;c_1$;
  \item $c_0\parallel c_1\sim (\mathbf{skip};c_0)\parallel c_1\sim c_1;c_0$.
\end{enumerate}
\end{lemma}

\begin{proof}
It is obvious by Proposition \ref{SePa2} and \ref{skipP2}.
\end{proof}

\subsection{Relations between Operational and Denotational Semantics}\label{pplrs}

The operational and denotational semantics still agree on the evaluation of $\mathbf{Aexp}$ and $\mathbf{Bexp}$, we do not repeat any more, please refer to \cite{FS} for details. We
will prove the agreement of the case $\mathbf{Com}$ as follows.

\begin{lemma}\label{lem}
For all commands $c$ and states $\sigma,\sigma'$,

$$\langle c,\sigma\rangle\rightarrow\sigma' \Rightarrow (\sigma,\sigma')\in\mathcal{C}\sembrack{c}$$
\end{lemma}

\begin{proof}
We will use rule-induction on the operational semantics of commands. For $c\in\mathbf{Com}$ and $\sigma,\sigma'\in\Sigma$, define

$$P(c,\sigma,\sigma')\Leftrightarrow_{def}(\sigma,\sigma')\in\mathcal{C}\sembrack{c}$$

We will show $P$ is closed under the rules for the execution of commands, and we will only prove the new case of $\parallel$, other commands please refer to \cite{FS} for details.

Recall the transition rules of $\parallel$ are:

%

$$\frac{\langle c_1,\sigma\rangle\rightarrow\sigma'\quad \langle c_0,\sigma\rangle\rightarrow\sigma''}{\langle c_0\parallel c_1,\sigma\rangle\rightarrow\sigma'\uplus\sigma''}$$

Assume that

$$\langle c_0,\sigma\rangle \rightarrow \sigma' \& P(c_0,\sigma,\sigma')\&\langle c_1, \sigma\rangle \rightarrow \sigma'' \&P(c_1,\sigma,\sigma'')$$

From the meaning of $P$, we can get that

$$\mathcal{C}\sembrack{c_0}\sigma = \sigma'\textrm{ and }\mathcal{C}\sembrack{c_1}\sigma = \sigma''$$

We can get

$$\mathcal{C}\sembrack{c_0\parallel c_1}\sigma=\sigma'\uplus\sigma''$$

which means that $P(c_0\parallel c_1,\sigma,\sigma'\uplus\sigma'')$ holds for the consequence of the rule, and is closed under this rule.
\end{proof}

\begin{theorem}
For all commands $c$ and states $\sigma,\sigma'$,

$$\mathcal{C}\sembrack{c}=\{(\sigma,\sigma')|\langle c,\sigma\rangle\rightarrow \sigma'\}$$
\end{theorem}

\begin{proof}
Lemma \ref{lem} gives the $\Leftarrow$ direction of proof, we only need to prove

$$(\sigma,\sigma')\in\mathcal{C}\sembrack{c}\Rightarrow \langle c,\sigma\rangle\rightarrow \sigma'$$

It is sufficient to induct on the structure of command $c$, we only prove the new case of $c\equiv c_0\parallel c_1$, other cases please refer to \cite{FS} for details.

Suppose $(\sigma,\sigma'\uplus\sigma'')\in\mathcal{C}\sembrack{c}$. Then there are some states, such that $(\sigma,\sigma')\in\mathcal{C}\sembrack{c_0}$, $(\sigma,\sigma'')\in\mathcal{C}\sembrack{c_1}$. By the hypothesis of $c_0,c_1$, we get

$$\langle c_0,\sigma\rangle\rightarrow \sigma'\textrm{ and }\langle c_1,\sigma\rangle\rightarrow \sigma''$$

So, $\langle c_0\parallel c_1,\sigma\rangle\rightarrow \sigma'\uplus\sigma''$, as desired.
\end{proof}

\subsection{Axiomatic Semantics}\label{pplas}

In this section, we give an axiomatic semantics for PPL by extending the Hoare rules with parallelism.

\subsubsection{Extended Hoare Rules for Parallelism}{\label{ehr}}

PPL should be extended to support assertion.

For $\mathbf{Aexp}$, it should be extended to:

$$a::=n\quad |\quad X\quad|\quad i\quad|\quad a_0+a_1\quad|\quad a_0-a_1\quad|\quad a_0\times a_1$$

where $i$ ranges over integer variables, $\mathbf{Intvar}$.

For $\mathbf{Bexp}$, it should be extended to support boolean assertion:

$$A::=\mathbf{true} | \mathbf{false}| a_0=a_1| a_0\leq a_1| \neg A| A_0\wedge A_1| A_0\vee A_1| A_0\Rightarrow A_1| \forall i.A| \exists i.A$$

And the formation rule of $\mathbf{Com}$ is maintained:

$$c::=\mathbf{skip}\quad |\quad X:=a\quad|\quad c_0;c_1\quad|\quad
\textbf{if } b\textbf{ then } c_0\textbf{ else } c_1\quad|\quad \textbf{while } b\textbf{ do } c\quad|\quad c_0\parallel c_1$$

Note that, $\mathbf{Com}$ contains a parallel composition $\parallel$.

The denotational semantics should also contain an interpretation $I$.

The full extended Hoare rules are as follow.

Rule for $\mathbf{skip}$:

$$\{A\}\mathbf{skip}\{A\}$$

Rule for assignments:

$$\{B[a/X]\}X:=a\{B\}$$

Rule for sequencing:

$$\frac{\{A\}c_0\{C\}\quad\{C\}c_1\{B\}}{\{A\}c_0;c_1\{B\}}$$

Rule for conditionals:

$$\frac{\{A\wedge b\}c_0\{B\}\quad\{A\wedge\neg b\}c_1\{B\}}{\{A\}\textbf{if }b\textbf{ then }c_0\textbf{ else }c_1\{B\}}$$

Rule for while loops:

$$\frac{\{A\wedge b\}c\{A\}}{\{A\}\textbf{while }b\textbf{ do }c\{A\wedge\neg b\}}$$

Rule for consequence:

$$\frac{\models (A\Rightarrow A')\quad \{A'\}c\{B'\}\quad\models (B'\Rightarrow B)}{\{A\}c\{B\}}$$

Rule for parallelism:

$$\frac{\{A\}c_0\{C\}\quad\{C\}c_1\{B\}\quad \{A\}c_1\{D\}\quad \{D\}c_0\{B\}}{\{A\}c_0\parallel c_1\{B\}}$$

\subsubsection{Soundness of The Extended Hoare Rules}\label{sou}

We can prove that each rule is sound by the following soundness theorem.

\begin{theorem}
Let $\{A\}c\{B\}$ be a partial correctness assertion, if $\vdash \{A\}c\{B\}$, then $\models \{A\}c\{B\}$.
\end{theorem}

\begin{proof}
It is sufficient to induct on the rule to prove each rule is valid. We only prove the new case of $\parallel$ rule, other cases please refer to \cite{FS} for details.

Assume that $\models\{A\}c_0\{C\}$ and $\models \{C\}c_1\{B\}$, and $\models\{A\}c_1\{D\}$ and $\models \{D\}c_0\{B\}$. Let $I$ be an interpretation. Suppose $\sigma\models^I A$.
Then $\mathcal{C}\sembrack{c_0}\sigma\models^I C$ and $\mathcal{C}\sembrack{c_1}(\mathcal{C}\sembrack{c_0}\sigma)\models^I B$, and $\mathcal{C}\sembrack{c_1}\sigma\models^I D$ and
$\mathcal{C}\sembrack{c_0}(\mathcal{C}\sembrack{c_1}\sigma)\models^I B$. Hence, $\models \{A\}c_0\parallel c_1\{B\}$, as desired.
\end{proof}

\subsubsection{Completeness of The Extended Hoare Rules}\label{com}

G\"{o}del's Incompleteness Theorem implies that the extended Hoare rules are incomplete. We prove the relative completeness in the sense of Cook.

\begin{theorem}
PPL extended with assertion is expressive.
\end{theorem}

\begin{proof}
It is sufficient to induct on the structure of command $c$, such that for all assertions $B$ there is an assertion $w\sembrack{c,B}$, for all interpretations $I$

$$wp^I\sembrack{c,B}=w\sembrack{c,B}^I$$

We only prove the new case of parallelism $c\equiv c_0\parallel c_1$, other cases please refer to \cite{FS} for details.

Inductively define $w\sembrack{c_0\parallel c_1,B}\equiv w\sembrack{c_0,w\sembrack{c_1, B}}$ and $w\sembrack{c_0\parallel c_1,B}\equiv w\sembrack{c_1,w\sembrack{c_0, B}}$. Then, for
$\sigma\in\Sigma$ and any interpretation $I$,

$\sigma\in wp^I\sembrack{c_0\parallel c_1,B}$ iff $\mathcal{C}\sembrack{c_0\parallel c_1}\sigma\models^I B$

iff $\mathcal{C}\sembrack{c_1}(\mathcal{C}\sembrack{c_0}\sigma)\models^I B$ and $\mathcal{C}\sembrack{c_0}(\mathcal{C}\sembrack{c_1}\sigma)\models^I B$

iff $\mathcal{C}\sembrack{c_0}\sigma\models^I w\sembrack{c_1,B}$ and $\mathcal{C}\sembrack{c_1}\sigma\models^I w\sembrack{c_0,B}$

iff $\sigma\models^I w\sembrack{c_0,w\sembrack{c_1,B}}$ and $\sigma\models^I w\sembrack{c_1,w\sembrack{c_0,B}}$

iff $\sigma\models^I w\sembrack{c_0\parallel c_1,B}$.

\end{proof}

\begin{lemma}
For $c\in\mathbf{Com}$ and $B$ is an assertion, let $w\sembrack{c,B}$ be an assertion expressing the weakest precondition with $w\sembrack{c,B}^I=wp^I\sembrack{c,B}$. Then

$$\vdash\{w\sembrack{c,B}\}c\{B\}$$
\end{lemma}

\begin{proof}
It suffices to induct on the structure of commands $c$, we only prove the new case of parallelism $c\equiv c_0\parallel c_1$, other cases please refer to \cite{FS} for details.

For $\sigma\in\Sigma$ and any interpretation $I$,

$\sigma\models^I w\sembrack{c_0\parallel c_1,B}$ iff $\mathcal{C}\sembrack{c_0\parallel c_1}\sigma\models^I B$

iff $\mathcal{C}\sembrack{c_1}(\mathcal{C}\sembrack{c_0}\sigma)\models^I B$ and $\mathcal{C}\sembrack{c_0}(\mathcal{C}\sembrack{c_1}\sigma)\models^I B$

iff $\mathcal{C}\sembrack{c_0}\sigma\models^I w\sembrack{c_1,B}$ and $\mathcal{C}\sembrack{c_1}\sigma\models^I w\sembrack{c_0,B}$

iff $\sigma\models^I w\sembrack{c_0,w\sembrack{c_1,B}}$ and $\sigma\models^I w\sembrack{c_1,w\sembrack{c_0,B}}$.

We get $\vdash\{w\sembrack{c_0,w\sembrack{c_1,B}}\}c_0\parallel c_1\{B\}$ and $\vdash\{w\sembrack{c_1,w\sembrack{c_0,B}}\}c_0\parallel c_1\{B\}$.

Hence, by the consequence rule, we obtain

$$\vdash\{w\sembrack{c_0\parallel c_1,B}\}c_0\parallel c_1\{B\}$$
\end{proof}

\begin{theorem}
For any partial correctness assertion $\{A\}c\{B\}$, if $\models\{A\}c\{B\}$, then $\vdash\{A\}c\{B\}$.
\end{theorem}

\begin{proof}
Suppose $\models\{A\}c\{B\}$, then $\vdash \{w\sembrack{c,B}\}c\{B\}$ where $w\sembrack{c,B}^I=wp^I\sembrack{c,B}$ for any interpretation $I$ (by the above Lemma). Hence,
$\models (A\Rightarrow w\sembrack{c,B})$, we obtain $\vdash\{A\}c\{B\}$.
\end{proof}

\subsection{Non-determinism}\label{gc}

The guarded commands can make the use of non-determinism more rigorous. To provide each command with a conditional guard, it is useful to eliminate the uncontrolled non-determinism.

The syntax of guarded commands are also composed of $\mathbf{Aexp}$, $\mathbf{Bexp}$ and $\mathbf{Com}$, and the syntax of $\mathbf{Aexp}$ and $\mathbf{Bexp}$ are the same as those of PPL in section \ref{ppls}. And the formation rules for the command $c$ and guarded commands $gc$ are as follow.

$$c::=\mathbf{skip}\quad |\quad \mathbf{abort}\quad|\quad X:=a\quad|\quad c_0;c_1\quad|\quad
\textbf{if } gc\textbf{ fi}\quad|\quad \textbf{do } gc\textbf{ od}$$

$$gc::=b\rightarrow c\quad|\quad gc_0\fatbar gc_1$$

where $gc_0\fatbar gc_1$ is the alternative construct of $gc_0$ and $gc_1$.

The operational rules of commands:

$$\langle \mathbf{skip},\sigma\rangle\rightarrow\sigma$$

$$\frac{\langle a, \sigma\rangle\rightarrow n}{\langle X:=a,\sigma\rangle\rightarrow\sigma[n/X]}$$

$$\frac{\langle c_0,\sigma\rangle\rightarrow\sigma'}{\langle c_0;c_1,\sigma\rangle\rightarrow\langle c_1,\sigma'\rangle}\quad \frac{\langle c_0,\sigma\rangle\rightarrow\langle c_0',\sigma'\rangle}{\langle c_0;c_1,\sigma\rangle\rightarrow\langle c_0';c_1,\sigma'\rangle}$$

$$\frac{\langle gc,\sigma\rangle\rightarrow\langle c,\sigma'\rangle}{\langle\textrm{if }gc\textrm{ fi},\sigma\rangle\rightarrow\langle c,\sigma'\rangle}$$

$$\frac{\langle gc,\sigma\rangle\rightarrow\mathbf{fail}}{\langle\textrm{do }gc\textrm{ od},\sigma\rangle\rightarrow\sigma}\quad\frac{\langle gc,\sigma\rangle\rightarrow\langle c,\sigma'\rangle}{\langle\textrm{do }gc\textrm{ od},\sigma\rangle\rightarrow\langle c;\textrm{do }gc\textrm{ od},\sigma'\rangle}$$

The operational rules of guarded commands:

$$\frac{\langle b,\sigma\rangle\rightarrow\mathbf{true}}{\langle b\rightarrow c,\sigma\rangle\rightarrow\langle c, \sigma\rangle}$$

$$\frac{\langle gc_0,\sigma\rangle\rightarrow\langle c,\sigma'\rangle}{\langle gc_0\fatbar gc_1,\sigma\rangle\rightarrow\langle c,\sigma'\rangle}\quad \frac{\langle gc_1,\sigma\rangle\rightarrow\langle c,\sigma'\rangle}{\langle gc_0\fatbar gc_1,\sigma\rangle\rightarrow\langle c,\sigma'\rangle}$$

$$\frac{\langle b,\sigma\rangle\rightarrow\mathbf{false}}{\langle b\rightarrow c,\sigma\rangle\rightarrow\mathbf{fail}}\quad\frac{\langle gc_0,\sigma\rangle\rightarrow\mathbf{false}\quad \langle gc_1,\sigma\rangle\rightarrow\mathbf{false}}{\langle gc_0\fatbar gc_1,\sigma\rangle\rightarrow\mathbf{fail}}$$

\subsection{Communications}\label{comm}

In this section, we extend communicating processes with the support for true concurrency.

The syntax of PPL are also composed of $\mathbf{Aexp}$, $\mathbf{Bexp}$, the names of communication channels $\alpha,\beta,\gamma\in \mathbf{Chan}$, and $\mathbf{Com}$, and the syntax of $\mathbf{Aexp}$ and $\mathbf{Bexp}$ are the same as those of PPL in section \ref{ppls}. And the formation rules for the command $c$ and guarded commands $gc$ are as follow.

$$c::=\mathbf{skip} | \mathbf{abort}| X:=a\quad| \alpha?X| \alpha!a| c_0;c_1|
\textbf{if } gc\textbf{ fi}| \textbf{do } gc\textbf{ od}| c_0\parallel c_1| c\setminus\alpha$$

$$gc::=b\rightarrow c\quad|\quad b\wedge \alpha?X\rightarrow c\quad|\quad b\wedge \alpha!a\quad|\quad gc_0\fatbar gc_1$$

where $gc_0\fatbar gc_1$ is the alternative construct of $gc_0$ and $gc_1$.

The operational rules of commands:

$$\langle \mathbf{skip},\sigma\rangle\rightarrow\sigma$$

$$\frac{\langle a, \sigma\rangle\rightarrow n}{\langle X:=a,\sigma\rangle\rightarrow\sigma[n/X]}$$

$$\langle \alpha?X,\sigma\rangle\xrightarrow{\alpha?n}\sigma[n/X]$$

$$\frac{\langle a,\sigma\rangle\rightarrow n}{\langle \alpha!a,\sigma\rangle\xrightarrow{\alpha!n}\sigma}$$

$$\frac{\langle c_0,\sigma\rangle\rightarrow\sigma'}{\langle c_0;c_1,\sigma\rangle\rightarrow\langle c_1,\sigma'\rangle}\quad \frac{\langle c_0,\sigma\rangle\rightarrow\langle c_0',\sigma'\rangle}{\langle c_0;c_1,\sigma\rangle\rightarrow\langle c_0';c_1,\sigma'\rangle}$$

$$\frac{\langle gc,\sigma\rangle\rightarrow\langle c,\sigma'\rangle}{\langle\textrm{if }gc\textrm{ fi},\sigma\rangle\rightarrow\langle c,\sigma'\rangle}$$

$$\frac{\langle gc,\sigma\rangle\rightarrow\mathbf{fail}}{\langle\textrm{do }gc\textrm{ od},\sigma\rangle\rightarrow\sigma}\quad\frac{\langle gc,\sigma\rangle\rightarrow\langle c,\sigma'\rangle}{\langle\textrm{do }gc\textrm{ od},\sigma\rangle\rightarrow\langle c;\textrm{do }gc\textrm{ od},\sigma'\rangle}$$

$$\frac{\langle c_0,\sigma\rangle\xrightarrow{\lambda}\langle c_0',\sigma'\rangle\quad c_0\%c_1}{\langle c_0\parallel c_1,\sigma\rangle\xrightarrow{\lambda}\langle c_0'\parallel c_1,\sigma'\rangle}$$

$$\frac{\langle c_1,\sigma\rangle\xrightarrow{\lambda}\langle c_1',\sigma'\rangle\quad c_0\%c_1}{\langle c_0\parallel c_1,\sigma\rangle\xrightarrow{\lambda}\langle c_0\parallel c_1',\sigma'\rangle}$$

$$\frac{\langle c_0,\sigma\rangle\xrightarrow{\lambda_1}\langle c_0',\sigma'\rangle\quad \langle c_1,\sigma\rangle\xrightarrow{\lambda_2}\langle c_1',\sigma''\rangle}{\langle c_0\parallel c_1,\sigma\rangle\xrightarrow{\{\lambda_1,\lambda_2\}}\langle c_0'\parallel c_1',\sigma'\uplus\sigma''\rangle}$$

$$\frac{\langle c_0,\sigma\rangle\xrightarrow{\alpha!n}\langle c_0',\sigma\rangle\quad \langle c_1,\sigma\rangle\xrightarrow{\alpha?n}\langle c_1',\sigma'\rangle}{\langle c_0\parallel c_1,\sigma\rangle\xrightarrow{\gamma_{\alpha}(n)}\langle c_0'\parallel c_1',\sigma'\rangle}$$

$$\frac{\langle c_0,\sigma\rangle\xrightarrow{\alpha?n}\langle c_0',\sigma'\rangle\quad \langle c_1,\sigma\rangle\xrightarrow{\alpha!n}\langle c_1',\sigma\rangle}{\langle c_0\parallel c_1,\sigma\rangle\xrightarrow{\gamma_{\alpha}(n)}\langle c_0'\parallel c_1',\sigma'\rangle}$$

$$\frac{\langle c,\sigma\rangle\xrightarrow{\lambda}\langle c',\sigma'\rangle}{\langle c\setminus\alpha,\sigma\rangle\xrightarrow{\lambda}\langle c'\setminus\alpha,\sigma'\rangle}\textrm{if }\lambda\equiv\alpha?n\textrm{ and }\lambda\equiv\alpha!n\textrm{ do not hold.}$$

Where $c_0\% c_1$ denotes that $c_0$ and $c_1$ are in race condition.

The operational rules of guarded commands:

$$\frac{\langle b,\sigma\rangle\rightarrow\mathbf{true}}{\langle b\rightarrow c,\sigma\rangle\rightarrow\langle c, \sigma\rangle}$$

$$\frac{\langle gc_0,\sigma\rangle\rightarrow\langle c,\sigma'\rangle}{\langle gc_0\fatbar gc_1,\sigma\rangle\rightarrow\langle c,\sigma'\rangle}\quad \frac{\langle gc_1,\sigma\rangle\rightarrow\langle c,\sigma'\rangle}{\langle gc_0\fatbar gc_1,\sigma\rangle\rightarrow\langle c,\sigma'\rangle}$$

$$\frac{\langle b,\sigma\rangle\rightarrow\mathbf{false}}{\langle b\rightarrow c,\sigma\rangle\rightarrow\mathbf{fail}}\quad\frac{\langle gc_0,\sigma\rangle\rightarrow\mathbf{false}\quad \langle gc_1,\sigma\rangle\rightarrow\mathbf{false}}{\langle gc_0\fatbar gc_1,\sigma\rangle\rightarrow\mathbf{fail}}$$

$$\frac{\langle b,\sigma\rangle\rightarrow\mathbf{false}}{\langle b\wedge \alpha?X\rightarrow c,\sigma\rangle\rightarrow\mathbf{fail}}$$

$$\frac{\langle b,\sigma\rangle\rightarrow\mathbf{false}}{\langle b\wedge \alpha!X\rightarrow c,\sigma\rangle\rightarrow\mathbf{fail}}$$

$$\frac{\langle b,\sigma\rangle\rightarrow\mathbf{true}}{\langle b\wedge \alpha?X\rightarrow c,\sigma\rangle\xrightarrow{\alpha?n}\langle c,\sigma[n/X]\rangle}$$

$$\frac{\langle b,\sigma\rangle\rightarrow\mathbf{true}\quad \langle a,\sigma\rangle\rightarrow n}{\langle b\wedge \alpha!a\rightarrow c,\sigma\rangle\xrightarrow{\alpha!n}\langle c,\sigma\rangle}$$

Note that, for true concurrency, we can see that communications, conflictions, and race conditions are solved as follows.

\begin{enumerate}
  \item Communication is explicitly supported in PPL, the two communicating commands $\alpha?X$ and $\alpha!X$ will merge to one communication command $\gamma_{\alpha}(X)$, and the unstructured communication will be eliminated;
  \item Since each command is with a guard, the conflictions among actions can be achieved by set the commands with exclusive guards;
  \item As the operational rules state, the actions in parallel in race condition must be executed sequentially and will cause the non-deterministic execution order. Though the execution order is non-deterministic, by setting appropriate guards to the parallel commands, the final execution configuration can be deterministic.
\end{enumerate}

We can get the following propositions. Where $\sim$ is an equivalence relation on commands by the definition, where $\Sigma$ is the set of states:

\begin{definition}[Equivalence of operational semantics for commands]
$c_0\sim c_1\textrm{ iff } \forall\sigma,\sigma'\in\Sigma, \langle c_0,\sigma\rangle\rightarrow \sigma'\Leftrightarrow\langle c_1,\sigma\rangle\rightarrow \sigma'$
\end{definition}

\begin{proposition}
$c_0\parallel c_1\sim c_1\parallel c_0$, for $c_0,c_1\in\mathbf{Com}$.
\end{proposition}

\begin{proof}
By use of the transition rules of $\parallel$, we can get the following derivations of $c_0\parallel c_1$ for $\forall \sigma\in\Sigma$:

$$\frac{\langle c_0,\sigma\rangle\xrightarrow{c_0}\sigma''\quad \langle c_1,\sigma''\rangle\xrightarrow{c_1}\sigma'}{\langle c_0\parallel c_1,\sigma\rangle\xrightarrow{c_0;c_1}\sigma'}$$

$$\frac{\langle c_1,\sigma\rangle\xrightarrow{c_1}\sigma'''\quad \langle c_0,\sigma'''\rangle\xrightarrow(c_0)\sigma'}{\langle c_0\parallel c_1,\sigma\rangle\xrightarrow(c_1;c_0)\sigma'}$$

$$\frac{\langle c_1,\sigma\rangle\xrightarrow{c_1}\sigma'\quad \langle c_0,\sigma\rangle\xrightarrow{c_0}\sigma''}{\langle c_0\parallel c_1,\sigma\rangle\xrightarrow{\{c_0,c_1\}}\sigma'\uplus\sigma''}$$

And we can get the following derivations of $c_1\parallel c_0$ for $\forall \sigma\in\Sigma$:

$$\frac{\langle c_1,\sigma\rangle\xrightarrow{c_1}\sigma'''\quad \langle c_0,\sigma'''\rangle\xrightarrow{c_0}\sigma'}{\langle c_1\parallel c_0,\sigma\rangle\xrightarrow{c_1;c_0}\sigma'}$$

$$\frac{\langle c_0,\sigma\rangle\xrightarrow{c_0}\sigma''\quad \langle c_1,\sigma''\rangle\xrightarrow{c_1}\sigma'}{\langle c_1\parallel c_0,\sigma\rangle\xrightarrow{c_0;c_1}\sigma'}$$

$$\frac{\langle c_0,\sigma\rangle\xrightarrow{c_0}\sigma'\quad \langle c_1,\sigma\rangle\xrightarrow{c_1}\sigma''}{\langle c_1\parallel c_0,\sigma\rangle\xrightarrow{\{c_0,c_1\}}\sigma'\uplus\sigma''}$$

So, it is obvious that $c_0\parallel c_1\sim c_1\parallel c_0$, for $c_0,c_1\in\mathbf{Com}$, as desired.
\end{proof}

\begin{proposition}
$(c_0\parallel c_1)\parallel c_2\sim c_0\parallel (c_1\parallel c_2)$, for $c_0,c_1,c_2\in\mathbf{Com}$.
\end{proposition}

\begin{proof}
By use of the transition rules of $\parallel$, we can get the following derivations of $(c_0\parallel c_1)\parallel c_2$ for $\forall \sigma\in\Sigma$:

$$\frac{\langle c_0,\sigma\rangle\xrightarrow{c_0}\sigma'\quad \langle c_1,\sigma\rangle\xrightarrow{c_1}\sigma''\quad \langle c_2,\sigma\rangle\xrightarrow{c_2}\sigma'''}{\langle (c_0\parallel c_1)\parallel c_2,\sigma\rangle\xrightarrow{\{c_0,c_1,c_2\}}\sigma'\uplus\sigma''\uplus\sigma'''}$$

And we can get the following derivations of $c_0\parallel (c_1\parallel c_2)$ for $\forall \sigma\in\Sigma$:

$$\frac{\langle c_0,\sigma\rangle\xrightarrow{c_0}\sigma'\quad \langle c_1,\sigma\rangle\xrightarrow{c_1}\sigma''\quad \langle c_2,\sigma\rangle\xrightarrow{c_2}\sigma'''}{\langle c_0\parallel (c_1\parallel c_2),\sigma\rangle\xrightarrow{\{c_0,c_1,c_2\}}\sigma'\uplus\sigma''\uplus\sigma'''}$$

So, it is obvious that $(c_0\parallel c_1)\parallel c_2\sim c_0\parallel (c_1\parallel c_2)$, for $c_0,c_1,c_2\in\mathbf{Com}$, as desired.

For the case of the parallel commands in race condition, we omit it.
\end{proof}

\begin{proposition}\label{SePa3}
For $c_0,c_1,c_2,c_3\in\mathbf{Com}$,

\begin{enumerate}
  \item $(c_0;c_1)\parallel c_2\sim (c_0\parallel c_2);c_1$;
  \item $(c_0;c_1)\parallel (c_2;c_3)\sim (c_0\parallel c_2);(c_1\parallel c_3)$.
\end{enumerate}

\end{proposition}

\begin{proof}
(1) By use of the transition rules of sequence and $\parallel$, we can get the following derivations of $(c_0;c_1)\parallel c_2$:

$$\frac{\langle c_0,\sigma\rangle\xrightarrow{c_0}\sigma'\quad \langle c_2,\sigma\rangle\xrightarrow{c_2}\sigma''}{\langle (c_0;c_1)\parallel c_2,\sigma\rangle\xrightarrow{\{c_0,c_2\}}\langle c_1,\sigma'\uplus\sigma''\rangle}$$

And we can get the following derivations of $(c_0\parallel c_2);c_1$:

$$\frac{\langle c_0,\sigma\rangle\xrightarrow{c_0}\sigma'\quad \langle c_2,\sigma\rangle\xrightarrow{c_2}\sigma''}{\langle (c_0\parallel c_2);c_1,\sigma\rangle\xrightarrow{\{c_0,c_2\}}\langle c_1,\sigma'\uplus\sigma''\rangle}$$

So, it is obvious that $(c_0;c_1)\parallel c_2\sim (c_0\parallel c_2);c_1$, for $c_0,c_1,c_2\in\mathbf{Com}$, as desired.

(2) By use of the transition rules of sequence and $\parallel$, we can get the following derivations of $(c_0;c_1)\parallel (c_2;c_3)$:

$$\frac{\langle c_0,\sigma\rangle\xrightarrow{c_0}\sigma'\quad \langle c_2,\sigma\rangle\xrightarrow{c_2}\sigma''}{\langle (c_0;c_1)\parallel (c_2;c_3),\sigma\rangle\xrightarrow{\{c_0,c_2\}}\langle c_1\parallel c_3,\sigma'\uplus\sigma''\rangle}$$

And we can get the following derivations of $(c_0\parallel c_2);(c_1\parallel c_3)$:

$$\frac{\langle c_0,\sigma\rangle\xrightarrow{c_0}\sigma'\quad \langle c_2,\sigma\rangle\xrightarrow{c_2}\sigma''}{\langle (c_0\parallel c_2);(c_1\parallel c_3),\sigma\rangle\xrightarrow{\{c_0,c_2\}}\langle c_1\parallel c_3,\sigma'\uplus\sigma''\rangle}$$

So, it is obvious that $(c_0;c_1)\parallel (c_2;c_3)\sim (c_0\parallel c_2);(c_1\parallel c_3)$, for $c_0,c_1,c_2,c_3\in\mathbf{Com}$, as desired.
\end{proof}

\begin{proposition}\label{skipP3}
$c\parallel \mathbf{skip}\sim c$, for $c\in\mathbf{Com}$.
\end{proposition}

\begin{proof}
By use of the transition rules of $\mathbf{skip}$ and $\parallel$, we can get the following derivations of $c\parallel \mathbf{skip}$:

$$\frac{\langle c,\sigma\rangle\xrightarrow{c}\sigma'\quad \langle \mathbf{skip},\sigma\rangle\rightarrow\sigma}{c\parallel \mathbf{skip},\sigma\rangle\xrightarrow{c}\sigma'\uplus\sigma}$$

And it is obvious that:

$$\frac{\langle c,\sigma\rangle\xrightarrow{c}\sigma'}{c,\sigma\rangle\xrightarrow{c}\sigma'}$$

For $\sigma'\uplus\sigma=\sigma'$, it is obvious that $c\parallel \mathbf{skip}\sim c$, for $c\in\mathbf{Com}$, as desired.
\end{proof}

\begin{lemma}\label{skipP23}
For  $c_0,c_1\in\mathbf{Com}$,

\begin{enumerate}
  \item $c_0\parallel c_1\sim c_0\parallel (\mathbf{skip};c_1)\sim c_0;c_1$;
  \item $c_0\parallel c_1\sim (\mathbf{skip};c_0)\parallel c_1\sim c_1;c_0$.
\end{enumerate}
\end{lemma}

\begin{proof}
It is obvious by Proposition \ref{SePa3} and \ref{skipP3}.
\end{proof}

From Lemma \ref{skipP23}, we can see that the execution orders of $c_0\parallel c_1$ cause non-determinism, they can be executed in any sequential order or in parallel simultaneously. But, with the assistance of guards, the final states after the execution of $c_0\parallel c_1$ can be deterministic.

\begin{proposition}\label{comlaws}
For $c,c_0,c_1\in\mathbf{Com}$,

\begin{enumerate}
  \item $\alpha!n\parallel \alpha?n\sim \gamma_{\alpha}(n)$;
  \item $(c;\alpha!n)\parallel \alpha?n\sim c;\gamma_{\alpha}(n)$;
  \item $(c;\alpha?n)\parallel \alpha!n\sim c;\gamma_{\alpha}(n)$;
  \item $(c_0;\alpha!n)\parallel (c_1;\alpha?n)\sim c_0\parallel c_1;\gamma_{\alpha}(n)$;
  \item $(c_0;\alpha?n)\parallel (c_1;\alpha!n)\sim c_0\parallel c_1;\gamma_{\alpha}(n)$.
\end{enumerate}
\end{proposition}

\begin{proof}
By use of the transition rules of $\parallel$, we can prove the above equations.
\end{proof}

From Proposition \ref{comlaws}, we can see that communications among parallel branches are eliminated and the parallelism is structured.

\subsection{Conflictions}\label{conf}

Corresponding to Figure \ref{EOTC3}, the program is:

$(1;(\textrm{if (b) then }(2;3)))\parallel (4;(\textrm{if (}\neg\textrm{ b) then }(5;6)))$

Corresponding to Figure \ref{EOTC4}. The program is:

$\textrm{if (b) then }(1;2;3)\parallel 4 \textrm{ else } 1\parallel (4;5;6)$

We can prove that the above two programs are equivalent, and the confliction between parallel branches is eliminated and the parallelism is structured.

\subsection{Structuring Algorithm}\label{sa}

By PPL, We know that the truly concurrent graph can be structured. As an implementation-independent language, the structuring algorithm of PPL can be designed as follows:

\begin{enumerate}
  \item Input the unstructured truly concurrent graph;
  \item By use of PPL, implement the graph as a program;
  \item By use of the laws of PPL, structure the program.
\end{enumerate}


\end{document}